\documentclass[10pt]{article}

\usepackage{amsfonts}
\usepackage{bm}
\usepackage{graphicx,subfigure}
\usepackage{epstopdf}
\usepackage{epsfig} 
\usepackage{mathptmx} 
\usepackage{mathrsfs}
\usepackage{amsthm}
\usepackage{amsmath,amssymb} 
\usepackage{cite}
\usepackage{color}
\usepackage{array}
\usepackage{floatflt}
\usepackage{lipsum}
\usepackage{amsfonts} 
\usepackage{lastpage}
\usepackage{latexsym}
\usepackage{booktabs}

\theoremstyle{plain}
\newtheorem{problem}{Problem}
\newtheorem{lemma}{Lemma}
\newtheorem{remark}{Remark}
\newtheorem{theorem}{Theorem}

\newcommand{\argmax}{\operatornamewithlimits{argmax}}

\usepackage[linesnumbered,ruled]{algorithm2e}
\usepackage[noend]{algorithmic}

\begin{document}

\title{Active Requirement Mining of Bounded-Time Temporal Properties of Cyber-Physical Systems}
\author{Gang Chen \and Zachary Sabato \and Zhaodan Kong}

\maketitle

\begin{abstract}
This paper uses active learning to solve the problem of mining bounded-time signal temporal requirements of cyber-physical systems or simply the requirement mining problem. By utilizing robustness degree, we formulates the requirement mining problem into two optimization problems, a parameter synthesis problem and a falsification problem. We then propose a new active learning algorithm called Gaussian Process Adaptive Confidence Bound (GP-ACB) to help solving the falsification problem. We show theoretically that the GP-ACB algorithm has a lower regret bound thus a larger convergence rate than some existing active learning algorithms, such as GP-UCB. We finally illustrate and apply our requirement mining algorithm on two case studies, the Ackley's function and a real world automatic transmission model. The case studies show that our mining algorithm with GP-ACB outperforms others, such as those based on Nelder-Mead, by an average of 30\% to 40\%. Our results demonstrate that there is a principled and efficient way of extracting requirements for complex cyber-physical systems. 
\end{abstract}

\section{Introduction}

In this paper, we propose the use of active learning for mining bounded-time signal temporal logic (STL) requirements of cyber-physical systems (CPSs). CPSs are characterized by the tight interaction of a collection of digital computing devices (the Cyber part) and a continuous-time dynamical system (the Physical part) \cite{lee2011introduction,alur2015CPS}. They are used for modeling in many safety-critical domains, such as automotive, medical, and aerospace industries, where the correctness of the end product is of significant importance \cite{jin2014powertrain,rajkumar2010cyber,jiang2012cyber}. However, in many industrial settings, the requirements intended to enforce the correctness guarantee are vague and in many cases expressed in natural languages, such as ``smooth steering'' and ``good fuel efficiency.'' This prevents the application of formal verification tools, such as Statistical Model Checking (SMC) \cite{zuliani2010bayesian}. Further, due to the complex nature of many CPSs, writing down the appropriate requirements to reflect the desirable system properties can be challenging even for experienced designers.

Given a system $S$, e.g., a Stateflow/Simulink model of a steering system, and a requirement template $\varphi$ with a set of unknown parameters $\theta$, the goal of this paper is to develop an improved method which can automatically infer a requirement $\varphi_{\theta}$ (often called a specification in formal methods literature) written in signal temporal logic, i.e., the system $S$ satisfies the requirement $\varphi_{\theta}$. Such a problem can be called either a \emph{requirement mining} problem or a \emph{temporal logic inference} problem. In the past few years with the introduction of the concept of \emph{robustness degree} \cite{donze2010robust,fainekos2009robustness}, the problem has received increased attention and has achieved significant progress. Robustness degree quantifies how strongly a given sample trajectory $s$ of a system $S$ exhibits a temporal logic property $\varphi_{\theta}$ as a real number rather than just providing a yes or no answer. A robustness degree of $r(s,\varphi_{\theta})$ means that the trajectory $s$ can tolerate perturbations up to the size of $|r(s,\varphi_{\theta})|$ and still maintain its current Boolean truth value with respect to $\varphi_{\theta}$.

With the help of robustness degree, the requirement mining problem of a CPS can be converted to an optimization problem for the expected robustness. Various techniques, such as particle swarm optimization \cite{haghighi2015spatel}, simulated annealing \cite{kong2014temporal}, and stochastic gradient descent algorithm \cite{kong2016TAC} can then be used to solve the optimization problem\footnote{Other optimization algorithms that have been used to solve similar optimization problems related to verification/falsification include rapidly exploring random tree (RRT) \cite{dreossi2015efficient}, cross-entropy method \cite{sankaranarayanan2012falsification}, Tabu search \cite{abbas2011linear}, and Ant Colony Optimization (ACO) \cite{annpureddy2011s}.}. Based on the nature of the search space, the requirement mining problems can be classified into two categories, \emph{parameter estimation} problems and \emph{structural inference} problems. In the former case, the structure of the requirement is given but with some unknown parameter values, for example, in a vehicular application, a requirement structure may be known as $F_{[0,\tau]} (v > \pi)$, which means that ``eventually between time 0 and some unspecified time $\tau$, the speed $v$ is greater than some unspecified value $\pi$.'' The goal of the parameter estimation problem is to identify appropriate values for the parameters \cite{bartocci2014data,jin2013mining}. In the latter example, even the structure of the requirement is unspecified. The structural inference problem is inherently hard. Candidate formulas have to be restricted to a certain template with which a partial order is well defined \cite{kong2014temporal,kong2016TAC,jones2014anomaly}. This paper focuses on the parameter estimation problem, but there are significant implications for structural inference as well.

Due to the hybrid nature of CPSs, the robustness degree function can be highly nonlinear. Furthermore, many CPSs are stochastic because of various uncertainties inherent to the system and its mathematical model. These complexities added together makes uniform sampling method inappropriate to solve the optimization problem (there are even cases in which the exact probability distribution over the parameter space is unknown a priori). Monte-Carlo techniques have been shown to be an effective sampling method to tackle the issue \cite{abbas2013probabilistic,jin2013mining}. It is worth pointing out that these techniques may suffer from slow convergence, meaning that the inference procedure may take a long time. In many applications, a quick verdict is needed, consider for instance an online diagnosis of a faulty safety-critical system.

In this paper, we use active learning to partly mitigate the need for a large number of iterations during optimization. The idea behind active learning is to accelerate convergence by actively selecting potentially ``informative'' samples, in contrast with random sampling from a predefined distribution \cite{Settles2010,ramdas2013algorithmic}. Active learning can be explained with a supervised video classification problem. With passive learning, a large amount of labeled data is needed for classification, e.g., cat videos vs. non cat videos. Labeling itself involves an oracle, e.g., a human, and it might be time consuming. Further, much of the data may be redundant, not contributing much to the classification. With active learning, based on the current information gathered about the data distribution, the learning algorithm only asks for the label of the most informative data. In our setting, the CPS model together with a stochastic model checker will serve as an oracle, providing the learner an approximate expected robustness degree based on the current most informative parameter. 

\paragraph{Contributions} The \emph{main} contribution of this paper is an active learning based scheme for requirement mining of cyber-physical systems. It unifies two complementary camps of requirement mining and verification philosophies: one is model based \cite{annpureddy2011s,abbas2013probabilistic,sankaranarayanan2012falsification,jin2013mining,asarin2012parametric}, and the other is data driven \cite{jones2014anomaly,kong2014temporal,bartocci2014data}. In our method, models are used as oracles, generating data which enables our method to gain knowledge of the system.  This helps focusing ongoing searches in promising parameter ranges, and thus eliminating unnecessary samples. The philosophy of our paper is similar to simulation based optimization \cite{gosavi2015simulation} but with a particular focus on utilizing machine learning techniques. \emph{Second}, instead of using existing learning algorithms, we develop a new active learning algorithm called Gaussian Process Adaptive Confidence Bound (GP-ACB). We prove that our GP-ACB algorithm converges faster than Gaussian Process Upper Confidence Bound (GP-UCP) algorithm \cite{srinivas2012information}, a state-of-the-art active learning algorithm. \emph{Third}, we integrate our GP-ACB algorithm with an existing verification tool, called Breach \cite{donze2010breach}. Using an automatic transmission controller as an example system, we show that the efficiency of Breach can be improved by as much as 40\%.

\paragraph{Organization} This paper is subdivided into the following sections. Section \ref{math} discusses the relevant background on signal temporal logic and Gaussian processes. Section \ref{problemForm} formally defines the requirement mining problem and shows how to transform the problem into an optimization problem with the help of robustness degree. Section \ref{active_learning} discusses our GP-ACB algorithm. Section \ref{caseStudies} provides two case studies to demonstrate our algorithm, an academic example with Ackley's function as the target function and an automotive example. Section \ref{conclusion} concludes the paper.

\section{Preliminaries}
\label{math}

\subsection{Signal Temporal Logic}
\label{sec:STL}
Given two sets $A$ and $B$, $\mathcal{F}(A,B)$ denotes the set of all functions from $A$ to $B$. Given a time domain $\mathbb{R}^+:=[0,\infty)$, a \emph{continuous-time, continuous-valued signal} is a function $s \in \mathcal{F}(\mathbb{R}^+,\mathbb{R}^n)$. We use $s(t)$ to denote the value of signal $s$ at time $t$, and $s[t]$ to denote the suffix of signal $s$ from time $t$, i.e., $s[t] = \{s(\tau)| \tau \geq t\}$.  

\em Signal temporal logic \em (STL) \cite{maler2004monitoring} is a temporal logic defined over signals.  STL is a predicate logic with interval-based temporal semantics. The  syntax of STL is defined as
\begin{equation}
\varphi := f(s) \sim d| \neg \varphi | \varphi_1 \wedge \varphi_2 | \varphi_1 \vee \varphi_2 | F_{[a,b)} \varphi | G_{[a,b)} \varphi,
\end{equation}
where $a$ and $b$ are non-negative finite real numbers, and  $f(s) \sim d$ is a predicate where $s$ is a signal, $f \in \mathcal{F}(\mathbb{R}^n,\mathbb{R})$ is a function, $\sim \in \{<,\geq\}$, and $d \in \mathbb{R}$ is a constant. The Boolean operators $\neg$ and $\wedge$ are negation (``not'') and conjunction (``and''), respectively. The other Boolean operators are defined as usual. The temporal operators $F$ and $G$ stand for ``Finally (eventually)'' and ``Globally (always)'', respectively.

The \em semantics \em of STL is recursively defined as
\begin{equation*}
 \begin{array}{rll}
s[t] \models (f(s)\sim d) & \text{ iff } & f(s(t)) \sim d \\
s[t] \models \varphi_1  \wedge \varphi_2  & \text{ iff } & s[t] \models \varphi_1 
\text{ and } s[t] \models \varphi_2 \\
s[t] \models \varphi_1  \vee \varphi_2  & \text{ iff } & s[t] \models \varphi_1 
\text{ or } s[t] \models \varphi_2 \\
s[t] \models  G_{[a,b)} \varphi  & \text{ iff } & \forall t' \in [t+a,t+b), s[t'] \models 
\varphi\\
s[t] \models F_{[a,b)} \varphi & \text{ iff } &
\exists t'
\in [t+a,t+b), \text{s.t. } s[t'] \models 
\varphi. 
 \end{array}
\end{equation*}
In plain English, $F_{[a,b)} \varphi$ means ``within $a$ and $b$ time units in the future, $\varphi$ is true'', and $G_{[a,b)} \varphi$ means ``for all times between $a$ and $b$ time units in the future $\varphi$ is true''.

STL is equipped with a \em robustness degree \em 
\cite{fainekos2009robustness,donze2010robust} (also called ``degree of 
satisfaction'') that quantifies how well a given signal $s$ satisfies a given formula $\varphi$.  The robustness is calculated recursively as follows
\begin{equation*}
 \begin{array}{rl}
 r(s,(f(s) < d),t) & = d-f(s(t)) \\
 r(s,(f(s) \geq d),t) & = f(s(t))-d \\
 r(s,\varphi_1 \wedge \varphi_2,t) &= \min \big(r(s,\varphi_1,t),r(s,\varphi_2,t) \big) \\
 r(s,\varphi_1 \vee \varphi_2,t) &= \max \big( r(s,\varphi_1,t),r(s,\varphi_2,t) \big) \\
  r(s, G_{[a,b)} \varphi,t) & = \underset{t' \in
[t+a,t+b)}{\min} r(s,\varphi,t') \\
r(s,F_{[a,b)} \varphi,t) & = \underset{t' \in
[t+a,t+b)}{\max}
r(s,\varphi,t').
 \end{array}
\end{equation*} 
We use $r(s,\varphi)$ to denote $r(s,\varphi,0)$. If $r(s,\varphi)$ is large and positive, then $s$ would have to change by a large deviation in order to violate $\varphi$.  

\emph{Parametric signal temporal logic} (PSTL) is an extension of STL where the bound $d$ and the endpoints of the time intervals $[a,b)$ are parameters instead of constants \cite{asarin2012parametric}. We denote them as \emph{scale} parameters $\pi=[\pi_1,...,\pi_{n_{\pi}}]$ and \emph{time} parameters $\tau=[\tau_1,...,\tau_{n_{\tau}}]$, respectively. A full parameterization is given as $[\pi, \tau]$. The syntax and semantics of PSTL are the same as those of STL. A \emph{valuation} $\theta$ is a mapping that assigns real values to the parameters appearing in an PSTL formula. A valuation $\theta$ of an PSTL formula $\varphi$ induces an STL formula $\varphi_{\theta}$. For example, if $\varphi = 
F_{[\tau_1,\tau_2)} (x < \pi_1)$ and $\theta([\pi_1,\tau_1,\tau_2]) = [0,0,3]$, then $\varphi_{\theta} =  F_{[0,3)} (x 
< 0)$.

\subsection{Gaussian Processes}
\label{sec:GP}

Formally, a Gaussian process (GP) is defined as a collection of random variables, any finite linear combination of which have a joint Gaussian distribution \cite{Rasmussen2006}. A simple example of a Gaussian process is a linear regression model $f(\vec{x}) = \phi(\vec{x})^T \vec{w}$, where $\vec{x} \in \mathbb{R}^n$, $\phi \in \mathcal{F}(\mathbb{R}^n,\mathbb{R}^N)$, which maps an $n$-dimensional $\vec{x}$ to an $N$-dimensional feature space (a feature is an individual, measurable property of a function space), and $\vec{w} \sim \mathcal{N}(\vec{0}, \Sigma)$, a zero mean Gaussian with covariance matrix $\Sigma$. Any GP is completely specified by its mean function $m(\vec{x})$ and its covariance function or kernel $k(\vec{x},\vec{x'})$
\begin{equation*}
\begin{array}{ll}
m(\vec{x})=E[f(\vec{x})],\\
k(\vec{x},\vec{x'})=E[(f(\vec{x})-m(\vec{x}))(f(\vec{x'})-m(\vec{x'}))].
\end{array}
\end{equation*}
As an example, for the linear regression model,
\begin{equation*}
\begin{array}{ll}
E[f(\vec{x})] = \phi(\vec{x})^T E[\vec{w}],\\
E[f(\vec{x})f(\vec{x'})]=\phi(\vec{x})^T E[\vec{w} \vec{w}^T] \phi(\vec{x'}) = \phi(\vec{x})^T \Sigma \phi(\vec{x'}).
\end{array}
\end{equation*}


A flat (or even zero) mean function $m(\vec{x})$ is chosen in the majority of cases in the literature. Such a choice does not cause many issues since the mean of the posterior process in not confined to zero. There is a large set of available kernels $k(\vec{x},\vec{x'})$. Two common ones, which are also used in this paper, are \cite{Anwar2015}
\begin{itemize}
\item Gaussian kernel with length-scale $l>0,k(\vec{x}_1,\vec{x}_2)=exp(-|\vec{x}_1-\vec{x}_2|^2/(2 l^2))$, where $|.|$ is the Euclidean length;
\item Mat\'{e}rn kernel with length-scale $l>0$ 
\begin{equation*}
k(\vec{x}_1,\vec{x}_2)=\frac{2^{1-\nu}}{\Gamma(\nu)}\left(\frac{\sqrt{2\nu}|\vec{x}_1-\vec{x}_2|}{\textit{l}}\right) ^{\nu}K_{\nu}\left( \frac{\sqrt{2\nu}|\vec{x}_1-\vec{x}_2|}{\textit{l}}\right),
\end{equation*}
where $K_{\nu}$ is the modified Bessel function and $\nu$ is a positive parameter.
\end{itemize}


\section{Requirement Mining Problem Formulation}
\label{problemForm}

In this section, we first provide some background on our cyber-physical system models. We then formally define the requirement mining problem. Finally, we show how to formulate the requirement mining problem as an optimization problem.


\subsection{Cyber-Physical Systems}
\label{sec:system}

In this paper, we study autonomous (closed-loop) cyber-physical systems. Notations from \cite{abbas2013probabilistic,sankaranarayanan2012falsification} are adopted here. A system $S$ maps an initial condition (or uncontrolled environmental conditions, e.g., road conditions) $\vec{x}_0 \in X_0 \subset \mathbb{R}^{n_x}$ to a discrete-time output signal $\vec{y} \in \mathcal{F}([0, T], Y)$ with $Y \subset \mathbb{R}^{n_y}$ and $T$ as the finite maximal simulation time. We assume both $X_0$ and $Y$ can be represented as the Cartesian product of intervals $[a_1, b_1] \times [a_2, b_2] \times \dots [a_n, b_n]$, where $a_i, b_i \in \mathbb{R}$.

\subsection{Problem Statement}
\label{sec:problem}

Formally, in this paper, we solve the following problem.
\begin{problem}
\label{original_problem}
Given a system $S$ with an initial condition set (or uncontrolled environmental condition set) $X_0 \subset \mathbb{R}^{n_x}$ and a parametric signal temporal logic formula $\varphi_{\theta}$ with unknown parameters $\theta \in \Theta \subset \mathbb{R}^{n_{\theta}}$, where $\Theta$ is the set of feasible valuation,  find a valuation $\theta$ such that 
\begin{equation}
\label{original requirment}
\forall \vec{x}_0 \in X_0: S(\vec{x}_0)[0] \models \varphi_{\theta}.
\end{equation}
That is to say, the output signal $\vec{y}$ of the system $S$, starting from any initial condition $\vec{x}_0 \in X_0$, satisfies $\varphi_{\theta}$ at time 0. We write $S(\vec{x}_0) \models \varphi_{\theta}$ as shorthand for $S(\vec{x}_0)[0] \models \varphi_{\theta}$.
\end{problem}


\subsection{Requirement Mining as Optimization}

Problem \ref{original_problem} is subject to the curse of dimensionality. If we solve Problem \ref{original_problem} directly, the dimension of the search space is $n_x+n_{\theta}$. In this paper, we use the following three methods to mitigate the computational complexity.
\begin{itemize}
\item To reduce the search space's dimension, we follow the idea described in \cite{jin2013mining}. We divide Problem \ref{original_problem} into two sub-problems: a parameter synthesis problem (Problem \ref{ff:synthesis_problem}) having a dimension of $n_{\theta}$ and a falsification problem (Problem \ref{falsificaton_problem}) having a dimension of $n_x$.
\item We then develop an active learning algorithm to progressively reduce the volume of the search space. The active learning algorithm focuses only on ``informative'' initial states, meaning it does not need to search the entire initial condition set $X_0$.
\item Finally, we take advantage of the monotonicity property of PSTL \cite{jin2013mining,kong2016TAC,kong2014temporal}. If a parameter has been found that solves Problem \ref{ff:synthesis_problem}, then there is no need to check for parameters that make the requirement more restrictive, as these parameters surely do not constitute an optimal solution.
\end{itemize}
The flowchart to solve Problem \ref{original_problem} is shown in Fig. 1.

\begin{figure}[!htbp]
\label{fig:Requirement mining flowchart}
\centering
\includegraphics[width=0.8\textwidth]{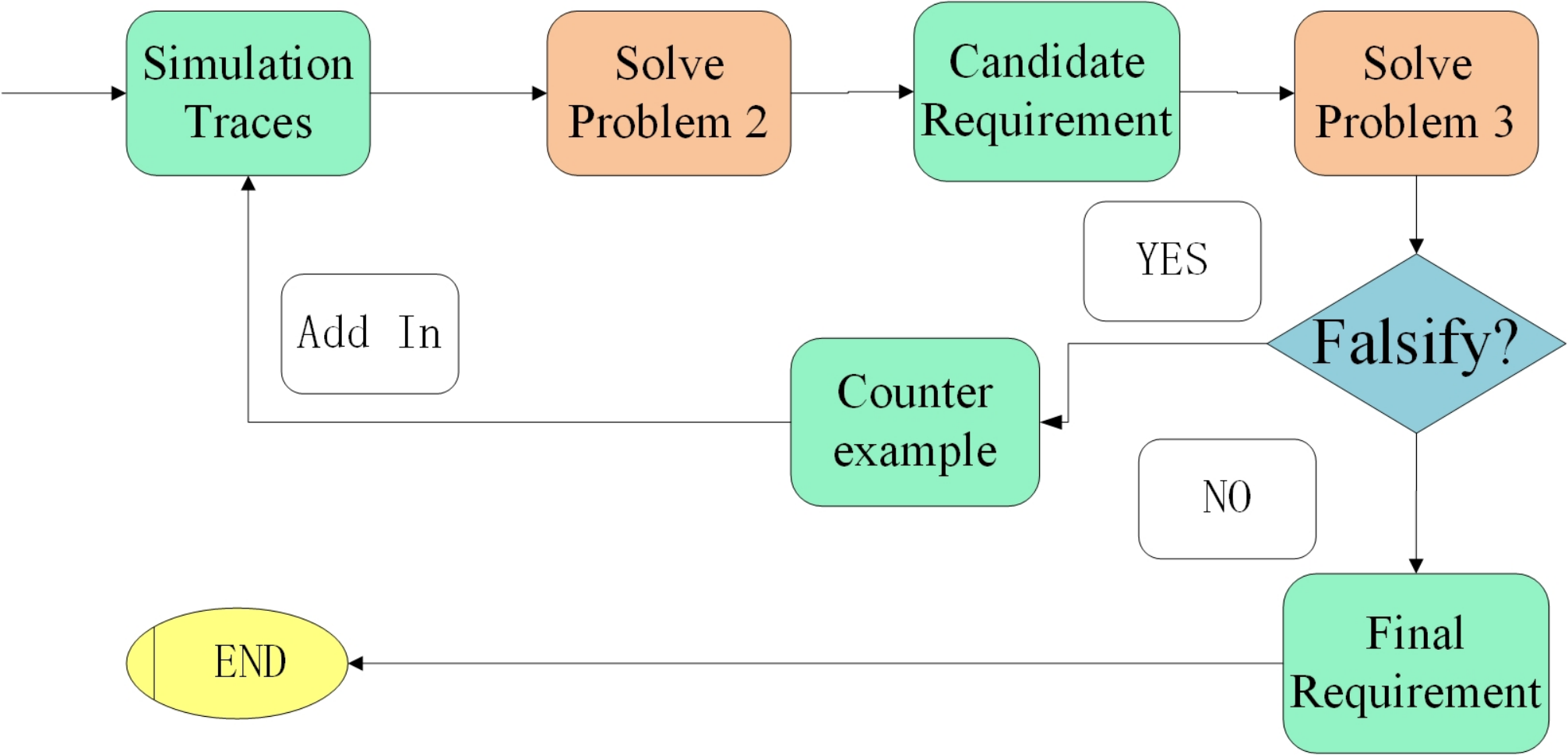}
\caption{Flowchart for solving Problem \ref{original_problem}. Our algorithm is based on the Breach toolbox \cite{donze2010breach,jin2013mining}. Our new active learning algorithm is used in the falsification step (shown in blue).}
\end{figure}

\begin{problem}
\label{ff:synthesis_problem}
\em Parameter Synthesis\em $\quad$ Given a system $S$ with a set of inputs $\bar{X}_0=\{\vec{x}_{0}^{i} \in X_0, i=1,2,\cdots,n_s\}$, where $n_s$ is the number of counter-example traces obtained from Problem \ref{falsificaton_problem}, and a PSTL formula $\varphi_{\theta}$ with unknown parameters $\theta \in \Theta \subset \mathbb{R}^{n_{\theta}}$, find a valuation $\theta$ to solve 
\begin{equation}
\max_{\theta}(0,\epsilon-\min_{\vec{x}_0 \in \bar{X}_0}(r(S(\vec{x}_0),\varphi_{\theta}))),
\end{equation}
where $\epsilon>0$ us a user-specified bound.
\end{problem}

\begin{problem}
\label{falsificaton_problem}
\em Falsification\em  $\quad$
Given a system $S$ with an initial condition set (or uncontrolled environmental condition set) $X_0 \subset \mathbb{R}^{n_x}$ and a PSTL formula $\varphi_{\theta}$ with particular parameters $\theta \in \Theta$, find an initial condition $\vec{x}_0 \in X_0$ to solve
\begin{equation}
\min_{\vec{x}_0 \in X_0}(r(S(\vec{x}_0),\varphi_{\theta})).
\end{equation}
\end{problem}

The max function $\max(0,\epsilon-\cdot)$ in Eqn. (3) is a modified hinge loss function. As the minimum of the robustness, $\min (r(S(\vec{x}_0),\varphi_{\theta}))$, is positive,  the loss function rewards values that are close to the bound 0 and at the same time positive. It is utilized here to tackle the issue related to the non-uniqueness of solutions to the requirement mining problem, as pointed out in \cite{jin2013mining}. For a particular $\theta$, the min functions in Eqn. (3) and Eqn. (4) reward initial states that lead to negative robustness degrees. Their goals are to find an initial state which leads to a trace that does not meet the requirement $\varphi_{\theta}$. The max-min function, Eqn. (3), then finds a parameter $\theta$, and, in turn, a requirement $\varphi_{\theta}$ such that for any initial state $x_0 \in X_0$, the output of the system has a positive robustness degree that is smaller than $\epsilon$. This means that the system satisfies the requirement $\varphi_{\theta}$, but only barely. Formally,

\begin{lemma}
\label{lemma problem}
A formula obtained by solving Problem 2 and Problem 3 together is a solution to Problem 1. I.e., if a formula $\varphi_{\theta^*}$, where 
\begin{equation*}
\theta^{*}=\argmax_{\theta}(0,\epsilon-\min_{\vec{x}_0}(r(S(\vec{x}_0),\varphi_{\theta})))
\end{equation*}
and furthermore $\forall \vec{x}_0 \in X_0, r(S(\vec{x}_0),\varphi_{\theta^{*}})> 0$, the statement (\ref{original requirment}) hold.
\end{lemma}

\begin{proof} \renewcommand{\qedsymbol}{}
Assume there is a parameter $\theta'$ obtained by solving Problem \ref{ff:synthesis_problem} and Problem \ref{falsificaton_problem}, and $\vec{x}_{0}^{'} \in X_0: S(\vec{x}_{0}^{'})[0] \nvDash \varphi_{\theta^{'}}$. According to Problem \ref{falsificaton_problem}, a $\vec{x}_{0}^{'} \in X_0$ that meets the requirements cannot be found. Thus, Lemma \ref{lemma problem} has been proven. 
\end{proof}

The most time-consuming part of solving Problem \ref{original_problem} is the generation of a trace given an initial state $\vec{x}_0$. For the parameter synthesis problem, there is no need to generate traces. Moreover, according to \cite{jin2013mining}, since PSTL satisfies the important property of monotonicity, the solution is quite efficient, via a binary search algorithm \cite{dreossi2015efficient}, for example. The falsification problem, on the other hand, does need a large number of traces being generated. Therefore, finding a way to minimize the number of traces needed in the falsification step will benefit the requirement mining process greatly. In this paper, we solve the falsification problem with a new active learning algorithm called Gaussian Process Adaptive Confidence Bound (GP-ACB). We will show theoretically (in Section \ref{active_learning}) and empirically (in Section \ref{caseStudies}) the effectiveness of our algorithm in reducing the computational time.

\section{Active Learning Algorithm}
\label{active_learning}

In the section, we propose an active learning algorithm, called Gaussian Process Adaptive Confidence Bound (GP-ACB) algorithm, to solve the falsification problem mentioned in Section 3. It is inspired by Gaussian process upper confidence bound (GP-UCB) \cite{srinivas2012information}. We analyze the regret bound of GP-ACB and show that it can achieve the same regret bound with a higher probability than GP-UCB.

\subsection{Gaussian Process Adaptive Confidence Bound Active Learning Algorithm}


Active learning algorithms are originally developed to solve classification problems when an oracle is needed to provide labels \cite{Settles2010}. The process of obtaining labels from the oracle can be expensive in terms of both time and money. Thus, the goal of any active learning algorithm is to achieve high classification or regression accuracy by using the fewest labeled instances. The requirement mining problem suffers from a similar issue as elaborated at the end of Section \ref{problemForm}. Our oracle is a simulator, e.g., a Stateflow/Simulink model. Given the complexity of many CPS models, to obtain a trace from the simulator can be costly in time. Thus, we need to decrease the number of simulations needed to learn a formula.  

One active learning algorithm is called Gaussian process upper confidence bound (GP-UCB) \cite{srinivas2012information}. At each step $t$, it solves the following problem
\begin{equation}
\label{eq:GP-UCB}
\vec{x}_t=\argmax_{\vec{x}\in D} m_{t-1}(\vec{x})+\beta_t^{\frac{1}{2}}\sigma_{t-1}(\vec{x}),
\end{equation}
where $D$ is the search space ($D=X_0$ for the requirement mining problem), $\beta_t$ is a function of $t$ and independent of $\vec{x}$ (an example of $\beta_t$ will be given later), $m_{t-1}(.)$ and $\sigma_{t-1}(.)$ are the mean and covariance function of the Gaussian process, respectively, and $\vec{x}_t$ is the instance that will be inquired at step $t$, meaning the label of $\vec{x}_t$ will be obtained from the oracle. The GP-UCB algorithm balances the classical exploitation-exploration trade-off by combining two strategies: the first term tends to pick those points that are expected to achieve high rewards (exploitation); and the second terms tend to pick those points that are uncertain (exploration). 

The second term of Eqn. (\ref{eq:GP-UCB}) only depends on the covariance function $\sigma(\vec{x})$, which can potentially make the exploration somewhat random and inefficient. To address this problem, we propose an algorithm called Gaussian Process Adaptive Confidence Bound (GP-ACB) by adding a normalization term $\eta_m(x)$ to Eqn. (\ref{eq:GP-UCB}) as follows:
\begin{equation}
\label{gpacb}
\vec{x}_t=\argmax_{\vec{x}\in D} m_{t-1}(\vec{x})+\eta_m(\vec{x})^{\frac{1}{2}}\beta_t^{\frac{1}{2}}\sigma_{t-1}(\vec{x}),
\end{equation}
where $\eta_m(\vec{x})$ normalizes the mean $m_{t-1}(\vec{x})$ and can be written explicitly as
\begin{equation*}
\eta_m(\vec{x})=\frac{m_{t-1}(\vec{x})-\min(m_{t-1}(\vec{x}))}{\max (m_{t-1}(\vec{x}))-\min (m_{t-1}(\vec{x}))}. 
\end{equation*}
It is quite obvious that $0\leq \eta_m(\vec{x})\leq 1$. $\eta_m(\vec{x})$ serves two main purpose here. First, it acts as an adaptive factor to uncertainty and favors exploration directions associated with increasing rewards. Thus, it can improve the exploration efficiency. Second, it provides an adaptive upper quantile of the marginal posterior $P(f(\vec{x})\mid\vec{y}_{t-1})$, where $\vec{y}_{t-1}$ is a vector containing all the observations until time $t-1$. In this paper, we assume that the observation at time $t$, $y_{t} \in \mathbb{R}$, is $y_{t}=f(\vec{x})+\varepsilon_t$ with $\epsilon_t\sim \mathcal{N}(0,\sigma^2)$ and $\sigma^2$ known. Pseudocode for the GP-ACB algorithm is provided in Algorithm 1.

\begin{algorithm} \label{algorithm:GP-ACB}
 \caption{GP-ACB Algorithm}
 \KwIn{\\ Search space $D$ ($X_0$ for the requirement mining problem); \\
GP priors $m(\vec{x})_0 = 0$ and $\sigma_0$;\\ 
Kernel function $k$;\\
Maximal simulation time $T$}
\begin{algorithmic}[1]
  \FOR{$i = 1 \text{ to } T$}
  \STATE $\text{Perform Bayesian update to obtain } m_{t-1}(\vec{x}) \text{ and } \sigma_{t-1}(\vec{x})$;
  \STATE $\text{Calculate the normalization factor } \eta_m(\vec{x})$;
  \STATE $\text{Choose } \vec{x}_t=\argmax_{\vec{x}\in D} m_{t-1}(\vec{x})+\eta_m(\vec{x})^{\frac{1}{2}}\beta_t^{\frac{1}{2}}\sigma_{t-1}(\vec{x})$;
  \STATE $\text{Calculate } \vec{y}_t=f(\vec{x}_t)+\varepsilon_t \text{ with } \epsilon_t\sim \mathcal{N}(0,\sigma^2)$.
    \ENDFOR
\end{algorithmic}
\end{algorithm}

\subsection{Regret Bound of GP-ACB}

The goal of any learning algorithm can be stated as follows: given an unknown reward function $f \in \mathcal{F}(D,\mathbb{R})$, maximize the sum of rewards $\sum_{t=1}^T f(\vec{x}_t)$, which is equivalent to finding a $\vec{x}^*$ such that $\vec{x}^*=\argmax_{\vec{x}\in D} f(\vec{x})$. A concept called regret bound can be used to quantify the convergence rate of a learning algorithm \cite{Niranjan2009,Settles2010,auer2009near}. First, the instantaneous regret at time $t$ is defined as $r_t = f(\vec{x}^*)-f(\vec{x}_t)$. Then, the cumulative regret $R_T$ after $T$ rounds is the sum of instantaneous regrets $R_T=\sum_{t=1}^T r_t$. A desired property of the learning algorithm is then to guarantee $\lim_{T \rightarrow \infty} R_T/T=0$, implying the convergence to the global maximum $\vec{x}^*$. Finally, the bounds on the average regret $R_T/T$ are directly related to the convergence rate of the learning algorithm. The lower the bound is, the faster the algorithm converges. This section investigates the regret bound of the GP-ACB algorithm.

Our proofs on the regret bound follow those in \cite{Niranjan2009}. Here we only consider the cases when the search space is finite, i.e., $|D| < \infty$.

\begin{lemma}
\label{Lemma 1}
Pick $\delta \in (0,1)$ and set $\beta_t=2\log(|D|\pi_t/\delta)$, where $\sum_{t\leq 1}\pi_t^{-1}=1$, $\pi_t>0$. Then,
\begin{equation*}
|f(\vec{x})-m_{t-1}(\vec{x})|\leq\eta_m(\vec{x})^{1/2}\beta_t^{1/2}\sigma_{t-1}(\vec{x}), \forall\vec{x} \in D \text{ and } \forall t\geq 1
\end{equation*}
holds with probability $\geq 1-\delta^{\eta_m(\vec{x})}$.
\end{lemma}

\begin{proof} \renewcommand{\qedsymbol}{}
For $\vec{x} \in D$ and $t\geq 1$. It is known that conditioned on $\vec{y}_{t-1}=(y_1,\cdots,y_{t-1})$, $\{\vec{x}_1,\cdots,\vec{x}_{t-1}\}$ are deterministic. Further, $f(\vec{x})\backsim\mathcal{N}(m_{t-1}(\vec{x}),\sigma_{t-1}^{2}(\vec{x}))$. Now if $r\backsim \textsl{N}(0,1)$, then
\begin{equation*}
\begin{array}{lr}
Pr{\lbrace r>c\rbrace}=e^{-c^{2}/2}(2\pi)^{-1/2}\int e^{-(\textsl{l}-c)^{2}/2-c(\textsl{l}-c)}dr\\
\leq e^{-c^{2}/2}Pr\lbrace r>0\rbrace=(1/2)e^{-c^{2}/2}.
\end{array}
\end{equation*}
for $c>0$, as $e^{-c(r-c)}\leq1$ for $r \geq c$. We have $Pr\lbrace|f(\vec{x})-m_{t-1}(\vec{x})|>\eta_m(\vec{x})\beta_t^{1/2}\sigma_{t-1}(\vec{x})\rbrace\leq e^{-\eta_m(\vec{x})\beta_t/2}$. Set $r=(f(\vec{x})-m_{t-1}(\vec{x}))/\sigma_{t-1}(\vec{x})$ and $c=\eta_m(\vec{x})^{1/2}\beta_t^{1/2}$. After applying the adaptive bound, we have
\begin{equation*}
|f(\vec{x})-m_{t-1}(\vec{x})|\leq\eta_m(\vec{x})^{1/2}\beta_t^{1/2}\sigma_{t-1}(\vec{x})  \quad\forall \vec{x} \in D
\end{equation*}
holds with probability $ \geq 1- |D|e^{-\eta_m(\vec{x})\beta_t/2}$. Choosing $|D|e^{-\eta_m(\vec{x})\beta_t/2}=\delta/\pi_t$, e.g., with $\pi_t=\pi^2t^2/6$, and using the adaptive bound for $t \in \mathbb{N}$, the statement holds.
\end{proof}

\begin{lemma}
\label{Lemma 2}
Fix $t\geq 1$, if $|f(\vec{x})-m_{t-1}(\vec{x})|\leq\eta_m(\vec{x})^{1/2}\beta_t^{1/2}\sigma_{t-1}(\vec{x})$, $\forall \vec{x} \in D$, then the regret $r_t$ is bounded by $2\beta_t^{1/2}\sigma_{t-1}(\vec{x}_t)$.
\end{lemma}

\begin{proof} \renewcommand{\qedsymbol}{}
According to the definition of $\vec{x}^{*}$, $m_{t-1}(\vec{x}_t)+\eta_m(\vec{x}_t)^{1/2}\beta_t^{1/2}\sigma_{t-1}(\vec{x}_t)$ \\ $\geq m_{t-1}(\vec{x}^{*})+ \eta_m(\vec{x}^{*})^{1/2}\beta_t^{1/2}\sigma_{t-1}(\vec{x}^{*})\geq f(\vec{x}^{*})$. Therefore, the instantaneous regret 
\begin{equation*}
\begin{array}{lll}
r_t=f(\vec{x}^{*})-f(\vec{x}_t)\\
\leq \eta_m(\vec{x}_t)^{1/2}\beta_t^{1/2}\sigma_{t-1}(\vec{x}_t)+m_{t-1}(\vec{x}_t)-f(\vec{x}_t)\\
\leq 2\eta_m(\vec{x}_t)^{1/2}\beta_t^{1/2}\sigma_{t-1}(\vec{x}_t)\leq 2\beta_t^{1/2}\sigma_{t-1}(\vec{x}_t)
\end{array}
\end{equation*}
\end{proof}

\begin{lemma}
\label{Lemma 3}
Pick $\delta \in (0,1)$ and set $\beta_t=2\log(\pi_t/\delta)$, where $\sum_{t\geq 1}\pi_t^{-1}=1, \pi_t>0$. Then,
\begin{equation*}
|f(\vec{x})-m_{t-1}(\vec{x})|\leq\eta_m(\vec{x})^{1/2}\beta_t^{1/2}\sigma_{t-1}(\vec{x})  \quad \forall t\geq 1
\end{equation*}
holds with probability $\leq 1-\delta^{\eta_m(\vec{x}_t)}$.
\end{lemma}

\begin{proof} \renewcommand{\qedsymbol}{}
For $\vec{x}\in D$ and $t\geq 1$. Conditioned on $\textbf{y}_{t-1}=\lbrace y_1,\cdots,y_{t-1}\rbrace,\lbrace\vec{x}_1,\cdots,\vec{x}_{t-1}\rbrace$ are deterministic. Further, $f(\vec{x}_t)\sim \mathcal{N}(m_{t-1}(\vec{x}),\sigma_{t-1}^2(\vec{x}))$. According to Lemma \ref{Lemma 1}, $Pr\lbrace|f(\vec{x}_t)-m_{t-1}(\vec{x}_t)|>\eta_m(\vec{x}_t)^{1/2}\beta_t^{1/2}\sigma_{t-1}(\vec{x}_t)\rbrace\leq e^{-\eta_m(\vec{x}_t)\beta_t/2}$. Since $e^{-\beta_t/2}=\delta/\pi_t$, and with the adaptive bound for $t \in \mathbb{N}$, the statement holds.
\end{proof}

\begin{lemma}
\label{Lemma 4}
Set $L_t=max(m_t(\vec{x}))-min(m_t(\vec{x}))$, $\forall \vec{x}\in D$, and let $\beta_t$ be defined as in Lemma \ref{Lemma 3}, then 
\begin{equation*}
1-\eta_m(\vec{x}_t)^{1/2}\leq \beta_t^{1/2}\sigma_{t-1}(\vec{x}_t)/L_t  \quad \forall t\geq 1
\end{equation*}
\end{lemma}

\begin{proof} \renewcommand{\qedsymbol}{}
Set $m_{t-1}(\vec{x}^m)=max(m_{t-1}(\vec{x}))$, $\forall \vec{x} \in D$, according to GP-ACB, $\eta_m(\vec{x}_t)^{1/2}\beta_t^{1/2}\sigma_{t-1}(\vec{x}_t)+m_{t-1}(\vec{x}_t)\geq\eta_m(\vec{x}^m)^{1/2}\beta_t^{1/2}\sigma_{t-1}(\vec{x}^m)+m_{t-1}(\vec{x}^m)$, then
\begin{equation*}
\begin{array}{lll}
m_{t-1}(\vec{x}^m)-m_{t-1}(\vec{x}_{t-1})\leq \eta_m(\vec{x}_t)^{1/2}\beta_t^{1/2}\sigma_{t-1}(\vec{x}_t)-\eta_m(\vec{x}^m)^{1/2}\beta_t^{1/2}\sigma_{t-1}(\vec{x}^m)\\
\Rightarrow 1-\eta_m(\vec{x}_t)^{1/2}\leq \frac{\beta_t^{1/2}}{L_t}(\eta_m(\vec{x}_t)^{1/2}\sigma_{t-1}(\vec{x}_t)-\eta_m(\vec{x}^m)^{1/2}\sigma_{t-1}(\vec{x}^m))\\
\qquad\qquad\quad  \leq\frac{\beta_t^{1/2}}{L_t}(\eta_m(\vec{x}_t)^{1/2}\sigma_{t-1}(\vec{x}_t)\leq \beta_t^{1/2}\sigma_{t-1}(\vec{x}_t)/L_t 
\end{array}
\end{equation*}
\end{proof}

\begin{remark}
\label{remark lemma 4}
Lemma \ref{Lemma 4} shows that when the scale of the function $L_t$ is large, $\eta_m(\vec{x})$ will be close to 1, meaning the GP-ACB algorithm degrades to GP-UCB. Conversely, when the scale of function is large, $m_t(\vec{x})$ will play a more important role in the search process, driving the algorithm to be very greedy. To investigate the effects of the scaling function on the optimization performance, a numerical experiment has been conducted and the results are shown in Section 5.1.
\end{remark}

Define $\gamma_{T}$ as the maximum information gain after $T$ rounds as follows \cite{Niranjan2009}:
\begin{equation*}
\gamma_{T} = \max_{T' \leq T} \frac{1}{2} \sum_{t=1}^{T'} \log (1+\sigma^{-2}\sigma_{t-1}^2(\vec{x}_t))
\end{equation*}
According to  \cite{Niranjan2009}, if the search space $D \in \mathbb{R}^d$ is compact and convex, where $d$ is the dimension of the search space, with the assumption that the kernel function satisfies $k(\vec{x},\vec{x}^{'}) \leqslant 1$, we have 
\begin{itemize}
\item $\gamma_T=\mathcal{O}((\log T)^{d+1})$ for Gaussian kernel;
\item $\gamma_T=\mathcal{O}\left( T^{d(d+1)/(2\nu+d(d+1))}(\log T)\right)$ for Mat\'{e}rn kernels with $\nu > 1$.
\end{itemize}
We can finally obtain the following bound for the GP-ACB algorithm.

\begin{theorem}
\label{theorem 1}
Let $\delta\in (0,1)$, $\beta_t=2\log(|D|t^2\pi^2/6\delta)$ , $m = \min_{t=(1,\cdots,T)} (\eta_m(\vec{x}_t))$ and $n = \max_{t=(1,\cdots,T)} (\eta_m(\vec{x}_t))$. Running GP-ACB results in a regret bound as follows
\begin{equation}
\label{eq:regret_bound}
Pr\lbrace R_T\leq\sqrt{nC_{1}T\beta_{T}\gamma_{T}}, \forall T\geq1\rbrace\geq 1-\delta^m,
\end{equation}
where $C_1=8/\log(1+\sigma^{-2})$.
\end{theorem}

\begin{proof} \renewcommand{\qedsymbol}{}
Define the information gain $I$ as follows:
\begin{equation*}
\label{information gain}
I(\textbf{y}_T,\textbf{f}_T)=\frac{1}{2}\sum_{t=1}^{T}\log(1+\sigma^{-2}\sigma_{t-1}^2(\vec{x}_t)),
\end{equation*}
where $\textbf{\textsl{f}}_T=(f(\vec{x}_1),\dots,f(\vec{x}_T))' \in \mathbb{R}^T$. According to Lemma $\ref{Lemma 1}$ and Lemma $\ref{Lemma 3}$, the regret bound $\lbrace r_t^2\leq 4\eta_m(\vec{x}_t)\beta_t\sigma_{t-1}^2(\vec{x}_t)$, $\forall t \geq 1\rbrace$ holds with probability $\geq 1-\delta^{\eta_m(\vec{x}_t)}\geq 1-\delta^m$. As $\beta_t $ is non-decreasing, we have 
\begin{equation}
\label{Regret bounds}
4\eta_m(\vec{x})\beta_t\sigma_{t-1}^2(\vec{x}_t)\leq 4n\beta_T\sigma^2(\sigma^{-2}\sigma_{t-1}^2(\vec{x}_t)\leq 4n\beta_T \sigma^2 S \log(1+\sigma^{-2}\sigma_{t-1}^2(\vec{x}_t))
\end{equation}
where $S=\sigma^{-2}/\log(1+\sigma^{-2})$, since $\sigma^{-2}\sigma_{t-1}^2(\vec{x}_t)\leq \sigma^{-2}k(\vec{x}_t,\vec{x}_t) \leq \sigma^{-2}$, \\$C_1=8/\log(1+\sigma^{-2})\geq 8\sigma^2$ and $h^2\leq S\log(1+h^2) for\quad h \in [0,\sigma^{-2}]$. As $C_1=8\sigma^2S$, for $T\geq 1$ we have
\begin{equation*}
\begin{array}{ll}
\sum_{t=1}^T r_t^2\leq\sum_{t=1}^T 4\eta_m(\vec{x})\beta_t\sigma_{t-1}^2(\vec{x}_t) \leq n\sum_{t=1}^T\frac{1}{2}\beta_T C_1 \log(1+\sigma^{-2}\sigma_{t-1}^2(\vec{x}_t))\\
\qquad\quad\leq nC_1\beta_T\gamma_T.
\end{array}
\end{equation*}
According to Cauchy-Schwarz inequality, $R_T^2\leq T\sum_{t=1}^T r_t^2$. Theorem \ref{theorem 1} has been proven.
\end{proof}

\begin{remark}
\label{remark for regret bounds}
The regret bound of the GP-UCB algorithm is \cite{Niranjan2009} 
\begin{equation*}
Pr\lbrace R_T\leq\sqrt{C_{1}T\beta_{T}\gamma_{T}}, \forall T\geq1\rbrace\geq 1-\delta.
\end{equation*}
With the same parameter setting, the regret bound of our GP-ACB algorithm is shown as Eqn. (\ref{eq:regret_bound}). Based on the bound of $\eta(\vec{x)}$ shown in Lemma \ref{Lemma 4} and the definition of the information gain $\gamma_T$, we know that $\sigma_{t-1}(\vec{x}_t)$ is close to 0 when $T$ is large, which means $m$ and $n$ in Theorem 1 are close to one when iteration $T$ is large and are close to zero when iteration $T$ is small. This indicates that GP-ACB algorithm is greedy at the beginning and degrades to GP-UCB algorithm when iteration $T$ is large. Since $0 < m,n \leqslant 1$, we can get the conclusion that the GP-ACB algorithm can get the same regret bound more efficient than that of the GP-UCB algorithm. The regret bound can be easily translated into convergence rate. The maximum $\max_{t \leq T} f(\vec{x}_t)$ in the first $T$ iterations is no further from $f(\vec{x}^*)$, where $\vec{x}^*$ is the global optimum, than the average regret $R_T/T$. Thus, compared with GP-UCB, on average, the GP-ACB algorithm has a higher or equal convergence rate. 
\end{remark}

\section{Case Studies}
\label{caseStudies}

To demonstrate the capability of GP-ACB, in this section, we first use it to solve a global optimization problem with Ackley's function as the target function. The performance of GP-ACP is compared with those of other active learning algorithms. We then use GP-ACB to solve a full scale requirement mining problem of an automatic transmission system, a benchmark used widely in CPS community \cite{jin2013mining,jin2014powertrain,sankaranarayanan2012falsification}. In both cases, GP-ACB outperforms other active learning algorithms as well as some other state-of-art optimization algorithms such as Nelder-Mead. 

\subsection{Global Optimization of Ackley's Function}

To verify the performance of the proposed GP-ACB algorithm, we compare the GP-ACB algorithm with four types of Gaussian-Process-based strategy: (i) GP-UCB active learning, (ii) Batch-greedy UCB active learning \cite{Desautels2014}, (iii) pure exploration, i.e., choosing points of maximum variance at each step, and (iv) pure exploitation or greedy, i.e., choosing points of maximum mean at each step. We use Ackley's function (shown in Fig. 2.(a)) as the target function. Its formula is as follows,
\begin{equation*}
f(x,y)=-20 e^{-0.2\sqrt{0.5(x^2+y^2)}}-e^{0.5(\cos(2\pi x)+\cos(2\pi y))}+e+20,
\end{equation*}
where $e$ is the observation noise with zero mean and variance $\sigma^{2}$ at 0.025. The search space $D=[-5,5]^{2}$ is randomly discretized into 1000 points. We run each algorithm for $T=58$ iterations with sampling time $\delta=0.1$. Since the global minimum of the Ackley's function $(x^*,y^*)$ is known (unknown to the learning algorithms though), for the $i$-th trial, if $(x_t^i,y_t^i)$ is the solution obtained by running the algorithm for $t$ iterations, then mean regret for the algorithm at time $t$ is $\bar{R}_t = \sum_{i=0}^{N_t} [f(x_t^i,y_t^i)-f(x^*,y^*)]/N_t$, where $N_t$ is the number of trials. In this case study, we set $N_t=1000$. Each trial is initialized randomly. 

\begin{figure}[!htb]
\label{fig: comparison of regret}
\centering
\subfigure[]{\includegraphics[width=0.5\textwidth]{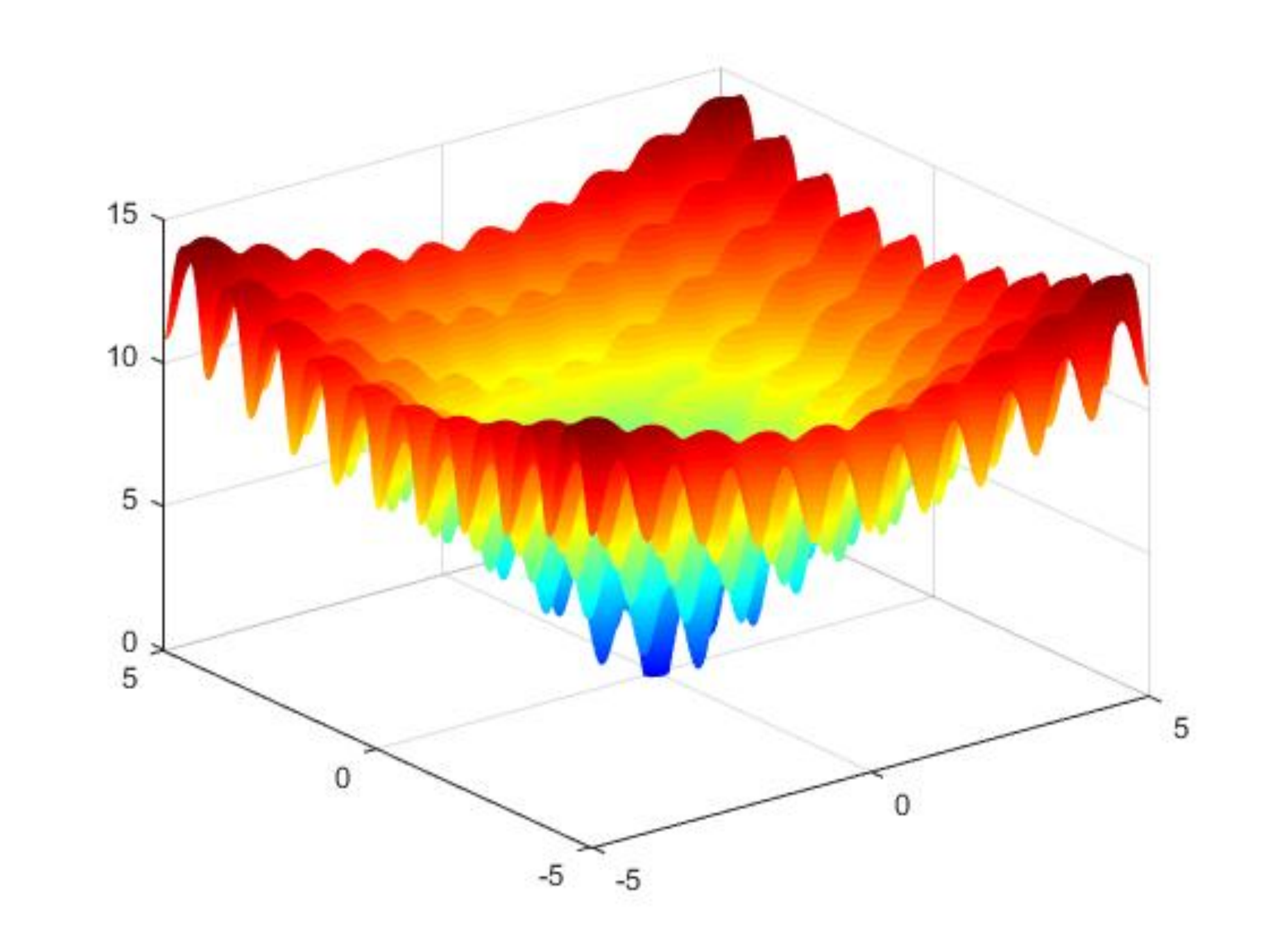}}\\
\subfigure[]{\includegraphics[width=0.48\textwidth]{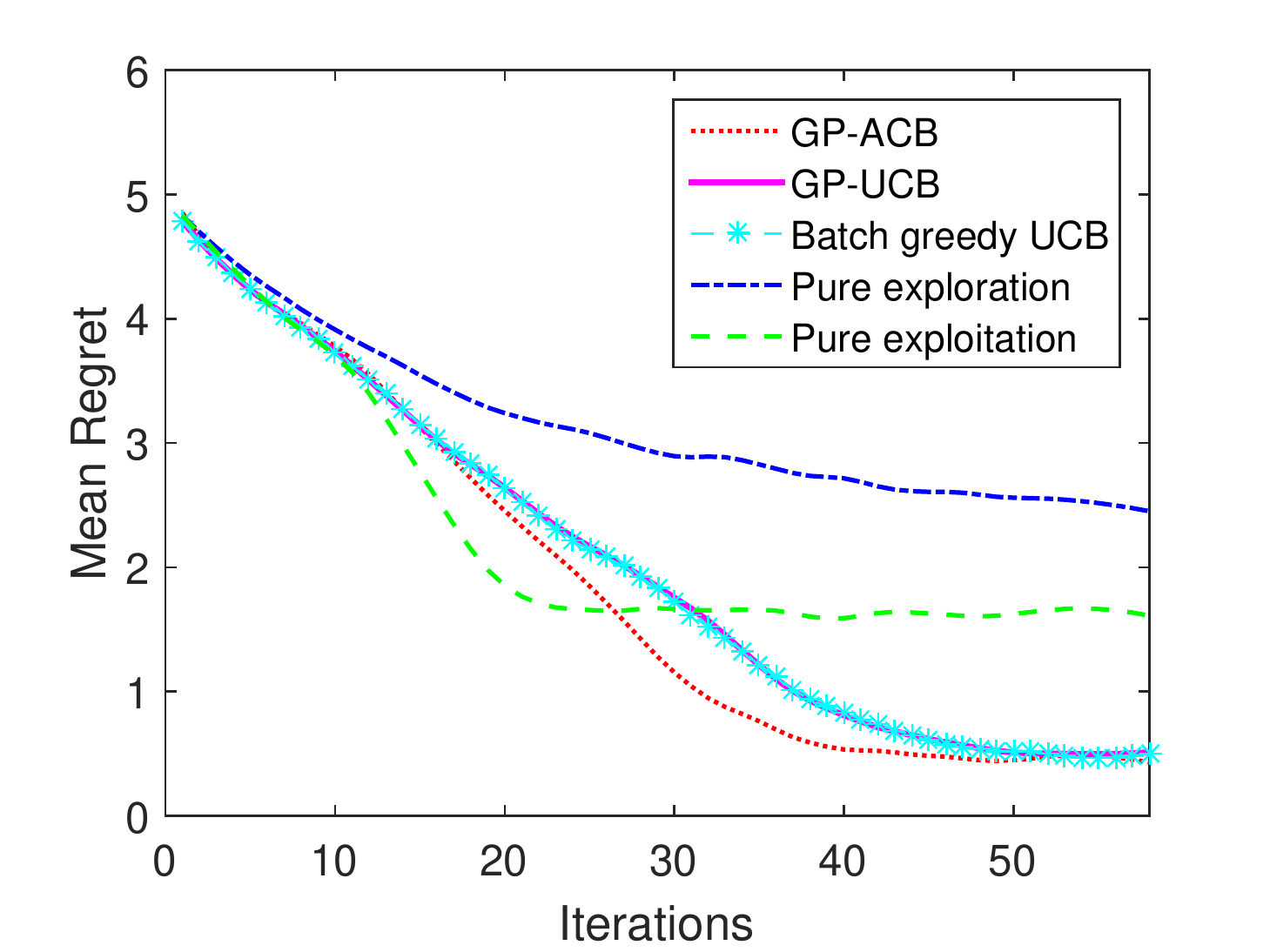}}
\subfigure[]{\includegraphics[width=0.48\textwidth]{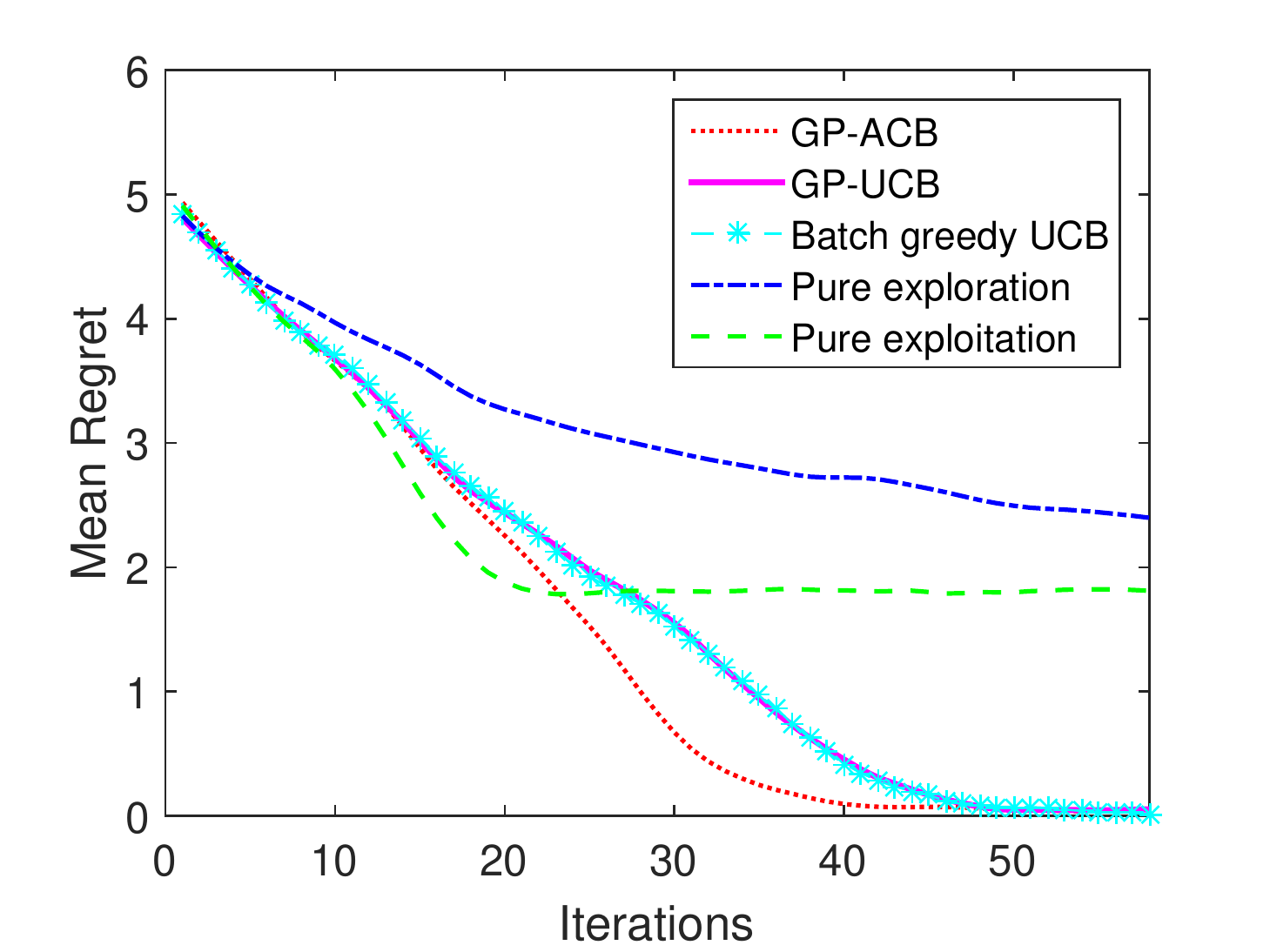}}
\caption{(a) Ackley's function. (b) and (c) compares the performance of GP-ACB with those of other strategies. The mean regret over 1000 trails $\bar{R}(t):= \bar{R}_t = \sum_{i=0}^{1000} [f(x_t^i,y_t^i)-f(x^*,y^*)]/1000$ is chosen as the performance metric. (b) shows the comparison results with Gaussian kernel. (c) shows the comparison results with Mat\'{e}rn kernel.}
\end{figure}

Fig. 2(b) and Fig. 2(c) show the comparison of the mean regret $\bar{R}_t$ incurred by the different Gaussian Process based algorithms with Gaussian kernel and Mat\'{e}rn kernel, respectively. With both kernels, GP-ACB outperforms the others. For instance, GP-UCB arrives at its minimum regret in an average of 58 iterations; while GP-ACB arrives at its minimum regret in an average of 45 iterations. Another interesting observation is that pure exploitation (greedy) strategy on average converges quicker than others. But the regret it converges to is on average higher than others, implying the convergence to local minimums. As for kernels, for this particular case, Mat\'{e}rn kernel outperforms the Gaussian kernel. This is not supervising, given that the Ackley's function (as shown in Fig. 2(a)) is quite ``non-smooth'' and Mat\'{e}rn kernel is designed to capture non-smoothness. 

\begin{figure}[!htb]
\label{fig:comparison of scaling factors}
\centering
\subfigure[]{\includegraphics[width=0.48\textwidth]{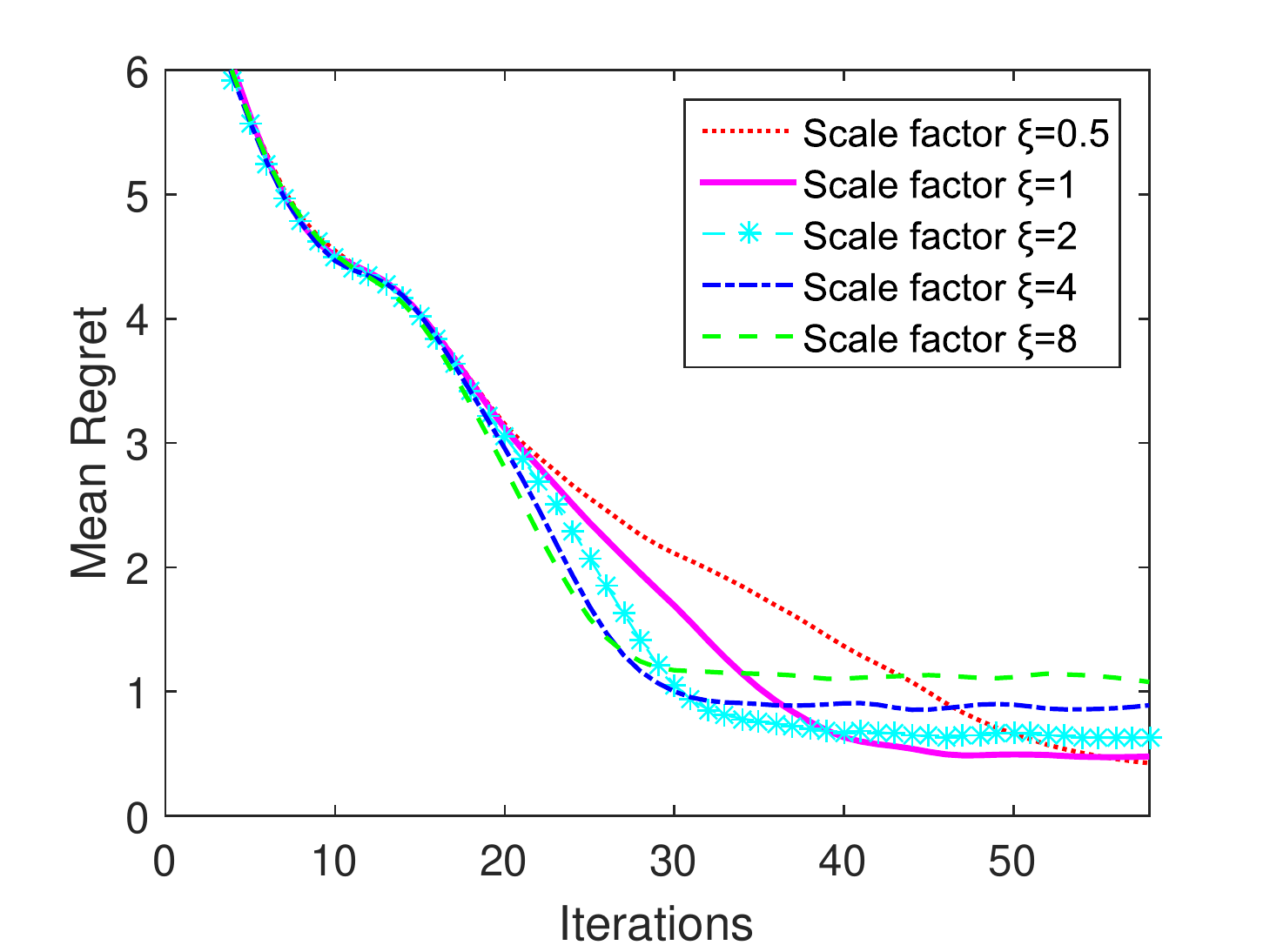}}
\subfigure[]{\includegraphics[width=0.48\textwidth]{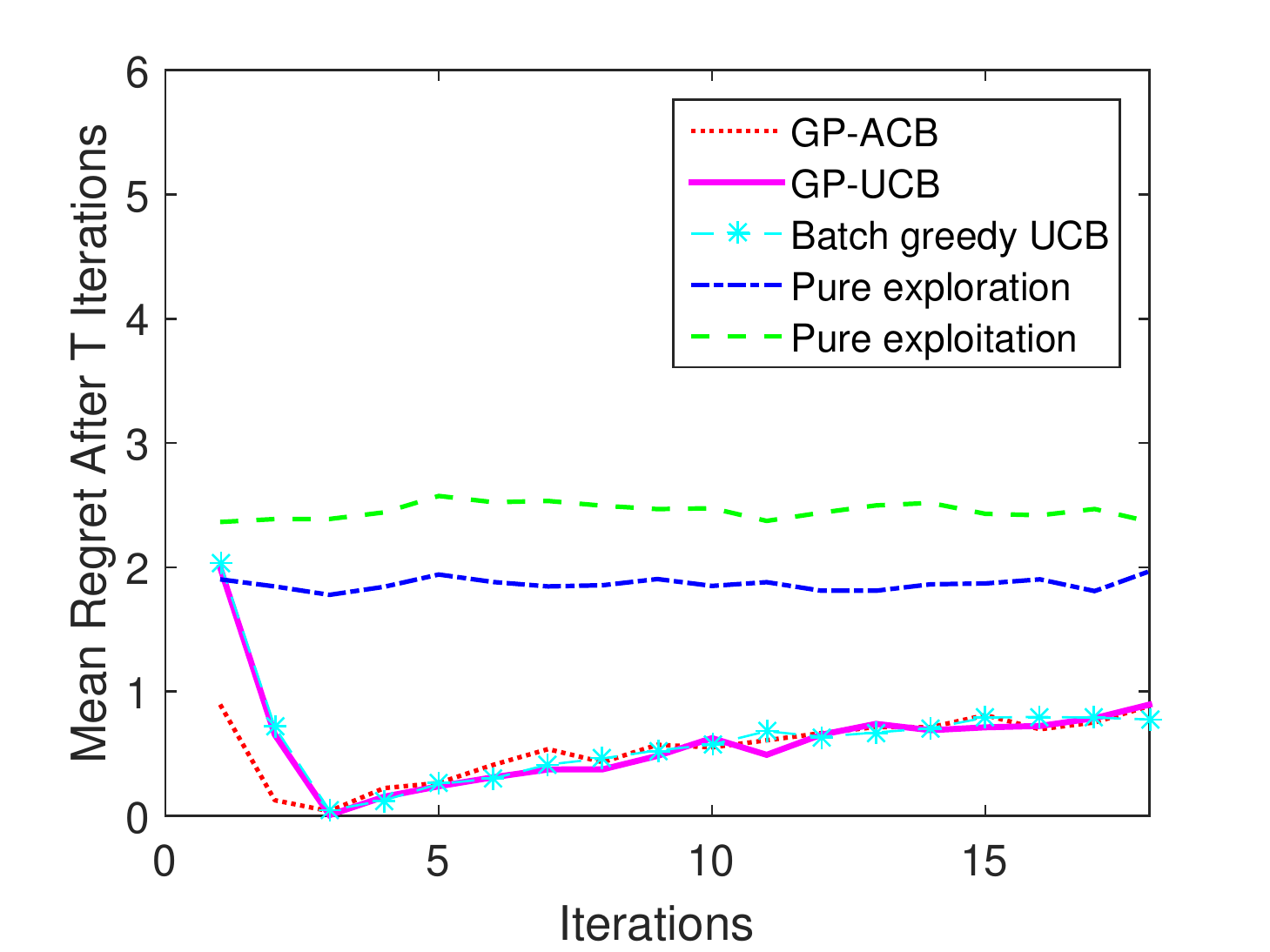}}
\caption{(a) The performance of the GP-ACB's algorithm with respect to the scaling factor $\xi$. The mean regret over 1000 trails $\bar{R}(t):= \bar{R}_t = \sum_{i=0}^{1000} [f(x_t^i,y_t^i)-f(x^*,y^*)]/1000$ is chosen as the performance metric. (b) The effects of the scaling factor $\xi$ on different algorithms. The mean regret at time $T$ over 1000 trails $\sum_{i=0}^{1000} [f(x_T^i,y_T^i)-f(x^*,y^*)]/1000$ is chosen as the performance metric. Here $T = 58$.}
\end{figure}

Remark \ref{remark lemma 4} shows that the range of the mean function $m(\vec{x})$ at time $t$, 
\begin{equation*}
L_t=max(m_t(\vec{x}))-min(m_t(\vec{x})),
\end{equation*}
affects the performance of the GP-ACB algorithm. When the range $L_t$ is large, the GP-ACB algorithm degrades to the GP-UCB algorithm. In this paper, we use a scale factor $\xi$ to change the range of the reward or cost function. Instead of optimizing $f(\vec{x})$ directly, we scale it by $\xi$ first and then solve the optimization problem with a new reward or cost function $\xi f(\vec{x})$. To demonstrate the effects, Fig. 3(a) shows the performance of GP-ACB with different scaling factors. It shows that when the scaling factor $\xi$ is large, the algorithm prefers exploitation over exploration and has the danger of trapping in a local optimum; when the scaling factor $\xi$ is small, the algorithm prefers exploration over exploitation and it may take more than necessary time to explore irrelevant regions. Thus, a proper $\xi$ needs to be selected to balance the trade-off between exploration and exploitation. Fig. 3(b) shows the effects of the scaling factor $\xi$ on different algorithms. It shows that when the scaling factor $\xi$ is large, the GP-UCB and GP-ACB algorithms have the same mean regret after 58 iterations, implying GP-ACB has degraded to GP-UCB. The pure exploration and pure exploitation algorithms' performance are not effected by the scaling factor $\xi$.

\subsection{Requirement Mining on an Automatic Transmission Model}

The model used in the following study is the same as the one used in \cite{jin2013mining}, a closed-loop model of a four-speed automatic transmission, shown in Fig 4. The model contains all necessary mechanical components, including engine, transmission and the longitudinal chassis dynamics. In the model, the throttle  position and brake torque are the input signals. With the current gear selection, the transmission ratio ($Ti$) can be computed through the transmission block, and the output torque can be obtained with the engine speed ($Ne$), gear status and transmission RPM.

\begin{figure}[!htb]
\label{system model}
\centering
\includegraphics[width=0.8\textwidth]{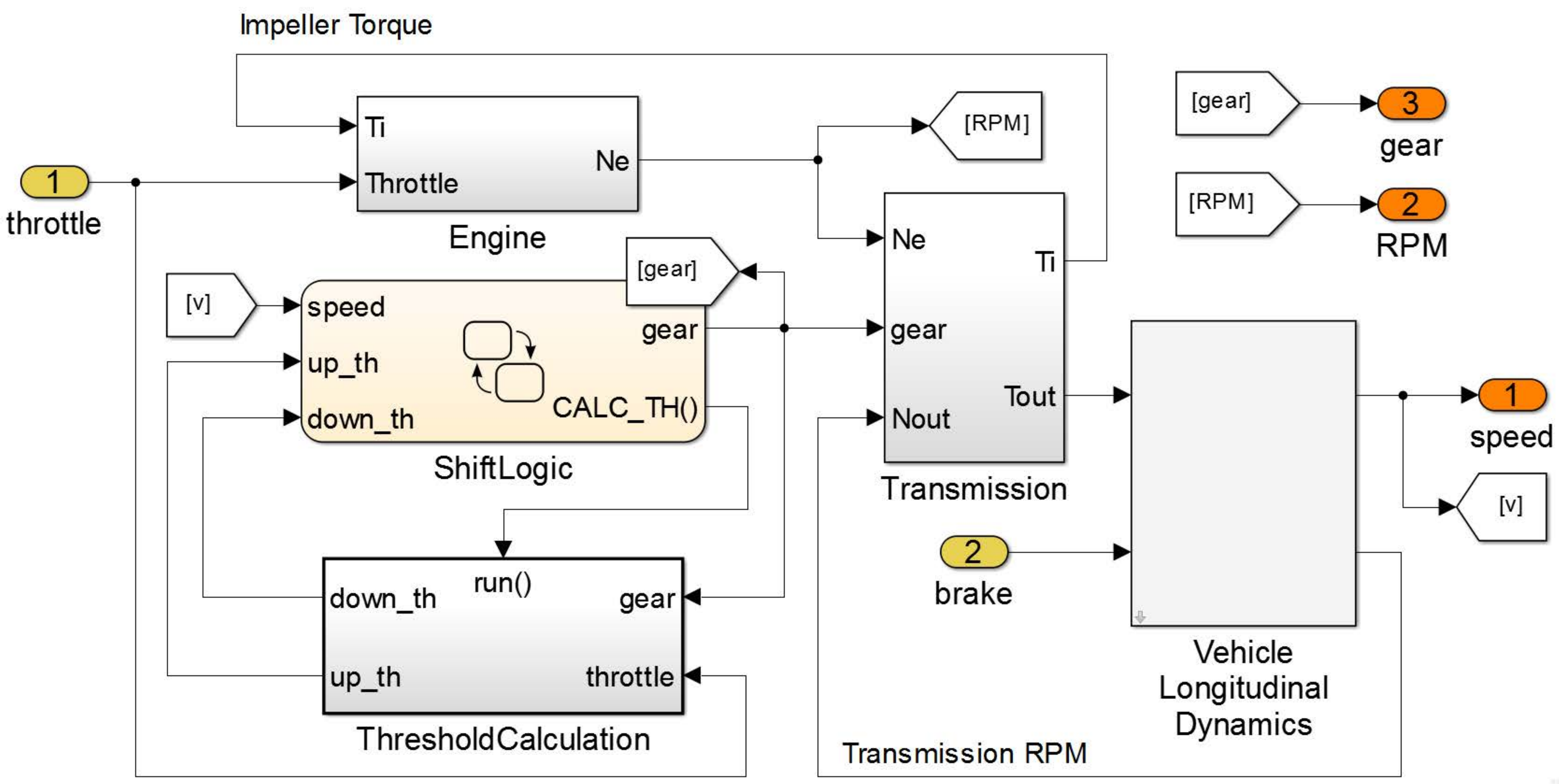}
\caption{Closed-loop Simulink model of an automatic transmission and controller, whose exogenous inputs are the throttle position and brake torque.}
\end{figure}

We tested three different template requirements which are the same with \cite{jin2013mining}, thus a comparison study can be conducted. The three template requirements are as follows:
\begin{enumerate}
\item Requirement $\varphi_{sp\_rpm}(\pi_1,\pi_2)$, which has a strong correlation with safety requirements characterizing the operating region for the engine parameters \textit{speed} and \textit{RPM}, specifying that always the \textit{speed} is below $\pi_1$ and \textit{RPM} is below $\pi_2$. 
\begin{equation*}
G((speed<\pi_1)\wedge(RPM<\pi_2)).
\end{equation*}
\item Requirement $\varphi_{rpm100}(\pi ,\tau)$, which measures the performance of the closed loop system. The value for $\tau$ specifies how fast the vehicle can reach a certain speed, and the value for $\pi$ can specify the lowest \textit{RPM} needed to reach the above speed. The formula specifies that the vehicle cannot reach the speed of 100 mph in $\tau$ seconds with \textit{RPM} always below $\pi$:
\begin{equation*}
\neg(F_{[0,\tau]}(speed>100)\wedge G(RPM<\pi)).
\end{equation*}
\item Requirement $\varphi_{stay}(\tau)$, which encodes undesirable transient shifting of gears, specifying that whenever the system shifts to gear 2, it dwells in gear 2 for at least $\tau$ seconds:
\begin{equation*}
G \left( \left( gear\neq 2\wedge F_{[0,\varepsilon]}gear=2\right) \Rightarrow G_{[0,\tau]}gear=2\right).
\end{equation*}
\end{enumerate}

Our requirement mining algorithm as shown in Fig. 1 and elaborated in Section 3.3 is implemented to mind requirements for the automatic transmission model. Our algorithm is based on the Breach toolbox \cite{donze2010breach} with the falsification problem solved by GP-ACB. The Breach toolbox uses the Nelder-Mead algorithm for falsification \cite{donze2010breach,jin2013mining}. Section 5.1 shows that the scaling factor affects the performance of the GP-ACB algorithm. We first run a few test trials in order to find the optimal scaling factor, the results of which are shown in Table 1. It can be seen that when scaling factor is set to 0.5, the time spent on falsification, the time spent on parameter synthesis and the number of simulations are all at their lowest values. Thus, the scaling factors are set to 0.5 for all the following experiments.  

\begin{table}[!htb]
\centering  
\label{scale factor}
\caption{Requirement mining results for GP-ACB with different scaling factors. In this table and Table 2, ``Parameter Values'' are the mined parameter values. For instance, parameter values $(4849, 155)$ together with the formula template $\varphi_{sp\_rpm}(\pi_1,\pi_2)$ means that the mined formula is $G((speed<4849)\wedge(RPM<155))$. ``Fals.(s)'' is the CPU time spent on solving the falsification problem. ``Synth.(s)'' is the CPU time spent on solving the parameter synthesis problem. ``$\#$Sim.'' is the number of simulations.``Robustness'' is the robustness degree corresponding to the mined parameters.}
\begin{tabular}{c|ccccc}   
\hline
Scaling Factor &\multicolumn{5}{c}{Requirement Mining Results}\\     
$\varphi_{sp\_rpm}(\pi_1,\pi_2)$ &Parameter Values  &Fals.(s)  &Synth.(s)  &$\#$Sim.  &Robustness\\ 
\hline  %
0.1   &(4849, 155) &112	&16.1	&393	&0.1541\\   
0.25  &(4846, 155)&109	&17.4	&385	&0.3510\\          
0.5   &(4849, 155)&67	&15.4	&255	&0.7108\\       
0.75  &(4846, 155)&79	&15.6	&276	&0.8197\\      
1     &(4847, 155)&107	&16.6	&367	&1.5098\\
\hline
\end{tabular}
\end{table}

The comparison results, averaged over 20 randomly started trails, for GP-ACB with Gaussian kernel, GP-ACB with Mat\'{e}rn kernel, GP-UCB and Nelder-Mead are shwon in Table 2. The results show GP-ACB with Mat\'{e}rn kernel outperforms the others in terms of the times spent on falsification, the times spent on parameter synthesis and the numbers of simulations. For instance, in mining the formula $\varphi_{sp\_rpm}(\pi_1,\pi_2)$, GP-UCB saves on average 10\% of simulation number compared with Nelder-Mead, while GP-ACB with Mat\'{e}rn kernel saves up to about 40\%. The same level of improvement is observed for the other two requirements, $\varphi_{rpm100}(\pi ,\tau)$ with a 40\% improvement and $\varphi_{stay}(\pi)$ with a 30\% improvement.

\begin{table}[!htb]
\centering  
\label{requirement mining}
\caption{Requirement mining results for GP-ACB with Gaussian kernel, GP-ACB with Mat\'{e}rn kernel, GP-UCB with Mat\'{e}rn kernel, and Nelder-Mead algorithm.}
\begin{tabular}{l|ccccc}   
\hline
Template &\multicolumn{5}{c}{Requirement Mining Results}\\     
$\varphi_{sp\_rpm}(\pi_1,\pi_2)$ &Parameter Values  &Fals.(s)  &Synth.(s)  &$\#$Sim. &Robustness\\ \hline  %
GP-ACB(Gaussian)         &(4845, 155) &100 &15.6 &330   &1.2812\\   
GP-ACB(Mat\'{e}rn) &(4849, 155) &67	&15.4	&255	&0.7108\\          
GP-UCB(Mat\'{e}rn) &(4844, 155)	&97	&15.4	&334	&0.7425\\      
Nelder-Mead        &(4857, 155)&112&13.3	&380	&0.6738\\ \hline
$\varphi_{rpm100}(\pi ,\tau)$ &Parameter Values  &Fals.(s)  &Synth.(s)  &$\#S$im. &Robustness\\ \hline  %
GP-ACB(Gaussian)         &(5997, 12.20) &64 &2.0 &207 &0.1556\\   
GP-ACB(Mat\'{e}rn) &(5997, 12.20) &59 &1.9 &179 &0.0802\\          
GP-UCB(Mat\'{e}rn) &(5997, 12.20) &63 &2.0 &190 &0.1116\\      
Nelder-Mead        &(5997, 12.20) &104 &1.4 &340 &0.06453\\ \hline
$\varphi_{stay}(\tau)$ &Parameter Values  &Fals.(s)  &Synth.(s)  &$\#$Sim. &Robustness\\ \hline  %
GP-ACB(Gaussian)         &0.0586 &296 &10.6 &941 &0.05\\   
GP-ACB(Mat\'{e}rn) &0.0586 &294 &10.8 &940 &0.05\\          
GP-UCB(Mat\'{e}rn) &0.0586 &363 &10.6 &1056 &0.05\\      
Nelder-Mead        &0.0586 &422 &9.3  &1246 &0.1\\ \hline
\end{tabular}
\end{table}

\section{Conclusions and Future Work}
\label{conclusion}

In this paper, we introduced an active learning method, called Gaussian Process adaptive confidence bound (GP-ACB), for mining requirements of bounded-time temporal properties of cyber-physical systems. The theoretical analysis of the proposed algorithm showed that it had a lower regret bound thus a higher convergence rate than other Gaussian-Process-based active learning algorithms, such as GP-UCB. By using two case studies, one of which was an automatic transmission model, we showed that our requirement mining algorithm outperformed other existing algorithms, e.g., those based on GP-UCB or Nelder-Mead, by an average of 30\% to 40\%. Our results have significant implications for not only the requirement mining but also the validation and verification of cyber-physical systems. We are currently exploring the possibility of utilizing active learning to solve the structural inference problem, i.e., to mind a requirement without any given template. 

                                  
\bibliographystyle{abbrv}
\bibliography{references}     

\end{document}